\documentclass[12pt, letterpaper]{article}

\usepackage[utf8]{inputenc}
\usepackage[english]{babel}

\usepackage[a4paper, total={6.3in, 8.9in}]{geometry}

\usepackage{tocloft}
\usepackage{mathptmx} 

\usepackage{verbatim}
\usepackage{enumitem}
\usepackage{array,multirow}
\usepackage{enumitem}
\usepackage{amssymb,amsfonts,amsmath,amsthm}
\usepackage{aliascnt} 
\usepackage{url}
\usepackage{tikz}
\usepackage{bbm}
\usepackage{accents}
\usepackage{thmtools}

\usetikzlibrary{arrows.meta}
\usetikzlibrary{calc}

\usepackage{xcolor}
\usepackage{ctable}
\usepackage{tikz}
\usetikzlibrary{matrix,arrows}
\usetikzlibrary{decorations.markings}

\definecolor{citec}{HTML}{324FDA} 
\definecolor{linkc}{HTML}{A0111A}
\definecolor{urlc}{HTML}{b55c87}
\usepackage[colorlinks,
    citecolor=citec,
    linkcolor=linkc,
    urlcolor=urlc,
    bookmarks = true,
    breaklinks,
]{hyperref}
\usepackage[nameinlink, noabbrev, capitalize]{cleveref}

\theoremstyle{plain}
\newtheorem{theorem}{Theorem}[section]

\newtheorem{lemma}[theorem]{Lemma}

\newtheorem{proposition}[theorem]{Proposition}
\newtheorem{claim}[theorem]{Claim}
\newtheorem{definition}[theorem]{Definition}

\newtheorem{remark}[theorem]{Remark}

\numberwithin{equation}{section}

\newcommand{\deff}{\triangleq}

\newcommand{\eps}{\varepsilon}

\def\poly{{\mathrm{poly}}}

\newcommand{\wt}{\widetilde}

\newcommand{\FF}{{\mathbb{F}}}

\newcommand{\EE}{{\mathbb{E}}}

\newcommand{\norm}[1]{\left\lVert #1 \right\rVert}

\newcommand{\EEE}{\mathcal{E}}

\begin{document}

\title{A Forward--Backward Weight Analysis of INW for Permutation Branching Programs}

\date{}

    \author{
  Gil Cohen\thanks{Tel Aviv University. \texttt{gil@tauex.tau.ac.il}. Supported by ERC starting grant 949499 and by the Israel Science Foundation grant 2989/24.  	
  }\and
  Dean Doron\thanks{Ben Gurion University. \texttt{deand@bgu.ac.il}. Supported in part by NSF-BSF grant 2022644.}
  \and
  Noam Goldgraber\thanks{Ben Gurion University and Tel Aviv University. \texttt{goldgrab@post.bgu.ac.il}. Supported by NSF-BSF grant 2022644.}
}

\maketitle

\begin{abstract}

We construct an $\eps$-error PRG for permutation read-once branching programs of length $n$ and width $w$ with seed length
\[
O\left((\log w+\log(1/\eps))\cdot \log n\right).
\]
This gives an exponential improvement in the dependence on $w$ compared with the constructions of De (CCC 2011) and Steinke (ECCC 2012). Compared with the work of Braverman, Rao, Raz, and Yehudayoff (FOCS 2010; SICOMP 2014), which applies more generally to regular branching programs and already achieves the optimal dependence on $w$, our result improves the dependence on the length $n$, attaining the optimal logarithmic dependence.

The generator itself is the classical INW PRG of Impagliazzo, Nisan, and Wigderson (STOC 1994).  We show that, for permutation branching programs, the INW generator can be instantiated with expanders whose degrees are polynomial in \(w\) and \(1/\eps\) and, crucially, independent of \(n\). To prove this, we analyze error propagation using program-dependent seminorms tailored to the branching program at hand. These seminorms build on the weight function introduced by Braverman et al. The key point is that, when measured in these adapted seminorms, the error \emph{does not accumulate} throughout the recursion.

Since our analysis relies only on the spectral expansion of the underlying expanders, our seed length tightly matches the recent lower bound for spectral analyses of the INW generator due to Hoza, Pyne, and Vadhan (Algorithmica 2024).
\end{abstract}

\thispagestyle{empty}

\newpage

\pagenumbering{gobble}

\setcounter{tocdepth}{2}
\tableofcontents

 \newpage
\pagenumbering{arabic}
\setcounter{page}{3}

\section{Introduction}

The $\mathbf{BPL}$ vs. $\mathbf{L}$ problem asks whether every polynomial-time randomized logarithmic-space algorithm can be simulated deterministically with only a constant-factor increase in space. The way that a space-bounded algorithm accesses its random bits can be modeled by a \emph{read-once branching program} (BP, for short; see \cref{def:bp}). A natural and well-studied approach to this problem---and, until very recently, the dominant route toward its resolution---is to construct an explicit \emph{pseudorandom generator} (PRG) with logarithmic seed length that fools polynomial-length, polynomial-width BPs.

Very recently, OpenAI claimed a proof of $\mathbf{BPL}=\mathbf{L}$ \cite{OpenAI}. Although the proposed proof does not proceed by constructing such a PRG, it seems to draw on ideas used in the construction of weighted pseudorandom generators \cite{BCG,PV21,CDRSTS21}, a relaxation of PRGs that allows signed weights rather than requiring a probability distribution. Constructing explicit logarithmic-seed PRGs for BPs in this regime remains a fundamental longstanding open problem.

The study of PRGs is motivated by more than deterministic simulation alone. PRGs support a useful algorithm-design methodology: first design an algorithm using randomness freely, which often makes the task simpler and more natural, and then use a PRG to reduce its randomness requirements. Furthermore, by replacing a long random tape with a short seed while retaining the algorithm's basic structure, this approach can preserve useful features of the original algorithm. More fundamentally, constructing PRGs asks how much randomness is needed to fool an entire class of computational tests, and thus seeks to understand the limitations of the class in question. This perspective also motivates the study of PRGs for classes such as $\mathbf{AC}^0$ circuits and low-degree polynomials, where the primary goal is to understand the limitations of these computational models rather than to derandomize algorithms. These motivations remain relevant independently of whether $\mathbf{BPL}=\mathbf{L}$.

PRGs for branching programs also play an important role in sketching and streaming algorithms, with applications to frequency-moment estimation \cite{Indyk06,KPTW23}. In these settings, a short seed provides a compact description of the randomness used by the algorithm, allowing random coefficients and hash values to be regenerated as needed rather than stored explicitly. Applications also extend beyond streaming, including derandomizing SDP rounding in approximation algorithms \cite{EIO02}, constructing succinct hash families with strong load-balancing guarantees \cite{CRSW11}, and constructing almost $k$-wise independent permutations \cite{KNR05} to name a few.


The seminal work of Nisan \cite{Nisan92} gave a PRG with seed length $O(\log^{2}n)$ for BPs over $n$ variables and width $w=\poly(n)$ -- the setting that is required for derandomizing $\mathbf{BPL}$ via standard derandomization. Impagliazzo, Nisan, and Wigderson \cite{INW94} subsequently introduced a second construction based on expander graphs. Although giving similar parameters to those in \cite{Nisan92}, the INW PRG turned
out to be more flexible, and will be the starting point for much of the discussion below.

Before we continue, let us briefly recall the structure of the INW generator. The construction is recursive: a generator for programs of length $m$ is used to construct a generator for programs of length $2m$. The generator is applied to the first half of the program, and an expander graph is then used to ``recycle'' its seed into a correlated seed for the second half. Since the recursion has depth $\log n$, the final seed length is $O(\log d \cdot \log n)$, where $d$ is the degree of the expanders. In the standard INW analysis, these expanders can be chosen to have degree $d$ which is polynomial in $n$, $w$, and $1/\eps$, yielding seed length
\[
    O\!\left((\log w+\log n+\log(1/\eps))\cdot \log n\right).
\]
In particular, in the $\mathbf{BPL}$ regime, this yields seed length $O(\log^2 n)$. 

Unfortunately, even for width $w = 4$, no construction is known with better than quadratic dependence on $\log n$. The only known exceptions are widths $2$ and $3$. Saks and Zuckerman showed that small-bias spaces fool width-$2$ BPs \cite{SZ95} (see also \cite{BDVY13}). The width-$3$ case, which is considerably more challenging, was treated by Meka, Reingold, and Tal \cite{MRT18}, who constructed a PRG with near-optimal dependence on $n$.

On the bright side, substantial progress has been made in constructing PRGs for more restricted classes of branching programs, a direction that has been studied extensively and successfully, with applications beyond the original $\mathbf{BPL}$ setting. A prominent example is the class of \emph{permutation} BPs, which has received considerable attention. These are branching programs in which, at every layer, both the $0$-transition and the $1$-transition are permutations of the state space. Note that such a computation is reversible. 

For this class, previous work achieved the optimal logarithmic dependence on the length $n$, and in fact, often using a substantially sharper analyses of the INW generator.
Kouck\'y, Nimbhorkar, and Pudl\'ak \cite{koucky2011} analyzed the INW generator for group products, a structured subclass of permutation branching programs. For general permutation BPs, their work achieved the optimal dependence on the length $n$, albeit with exponential dependence on the width $w$; more precisely, they obtained seed length
\[
    O\!\left((w!+\log(1/\eps))\cdot \log n\right).
\]
De \cite{De11} subsequently achieved polynomial dependence on $w$, obtaining seed length
\[
    O\!\left((w^8+\log(1/\eps))\cdot \log n\right).
\]
De's analysis uses representation-theoretic ideas. Steinke \cite{Steinke12} later gave a more elementary analysis, avoiding any reference to groups or representation theory, though still relying on rather intricate combinatorial arguments. This improved the width dependence to
\[
    O\!\left((w^4\log w+\log(1/\eps))\cdot \log n\right).
\]

This line of work shows that, for permutation BPs, the expanders underlying the INW generator can be chosen to have degree independent of the program length $n$, though exponential in the width $w$. Consequently, for constant width and constant error, the resulting seed length is  $O(\log n)$.

A more general class of branching programs is that of \emph{regular} BPs, in which every state has  in-degree $2$. Braverman, Rao, Raz, and Yehudayoff \cite{BRRY} gave a refined analysis of the INW generator for regular BPs, obtaining seed length
\[
    O\!\left((\log w+\log\log n+\log(1/\eps))\cdot \log n\right).
\]
BRRY showed that, for regular BPs, it suffices to use expanders whose degree is polynomial only in the recursion depth, namely in $\log n$, rather than in the program's length $n$; this is the source of the $\log\log n$ term in their seed length. Their argument is substantially simpler than the analyses of De and Steinke~\cite{De11,Steinke12}, which achieve optimal dependence on $n$ for permutation BPs, while also allowing the expander degree to depend only polynomially on $w$.

Nevertheless, the BRRY bound remains suboptimal in its dependence on the length. Within the INW framework, achieving seed length \(O(\log n)\) requires expander degrees that are completely independent of \(n\). Consequently, even a slowly growing loss across the recursion is prohibitive, and one must show that the relevant error does not accumulate at all.

\subsection{Our Result}

Despite a resurgence over the past five years in the study of PRGs and weighted PRGs \footnote{A weighted PRG is a relaxation of the standard notion of a PRG, introduced in~\cite{BCG}, in which each seed is assigned a real weight. These weights may be negative and need not be bounded in absolute value.}  for permutation BPs and their unbounded-width variants, as discussed in \cref{sec:other related work}, there had been no improvement for more than a decade on the original problem of constructing a PRG for permutation BPs, following the works of \cite{De11,Steinke12,BRRY}. In this work, we obtain the first such improvement by giving a new analysis of the INW generator. We show that it fools permutation BPs with seed length
\(
O\left((\log w+\log(1/\eps))\cdot \log n\right).
\)

Compared with the analyses of De~\cite{De11} and Steinke~\cite{Steinke12}, which also achieve the optimal dependence on \(n\), our result improves the dependence on \(w\) exponentially, attaining the optimal dependence on this parameter. Compared with the BRRY analysis~\cite{BRRY}, which applies more generally to regular BPs, our result achieves the optimal logarithmic dependence on the length \(n\).

\begin{theorem}[Main result; informal]\label{thm:informal-main}
    There exists an explicit PRG $G:\{0,1\}^s \rightarrow \{0,1\}^n$ that $\eps$-fools length-$n$ width-$w$ permutation branching programs, with seed length 
    \[
        s = O\!\left((\log w+\log(1/\eps))\cdot \log n\right).
    \]
    In particular, $G$ is the INW generator, properly instantiated.
\end{theorem}
For the complete formal statement, see \cref{thm:main-formal-result}.

Because our proof is a spectral analysis, the resulting seed length tightly matches the lower bound for spectral analyses established by Hoza, Pyne, and Vadhan~\cite{HWP24}, discussed below in \cref{sec:other related work}. Thus, we answer affirmatively the first open problem posed in their work.

\subsection{Other Related Work}\label{sec:other related work}

As noted above, the study of PRGs and weighted PRGs for permutation BPs, restricted subclasses, and related variants has seen significant activity in recent years. In this section, we survey the results most relevant to our work.

For restricted families of permutation BPs, improved seed lengths are known via generators that are not based on the INW construction. In a paper predating the definition of permutation BPs, Lovett, Reingold, Trevisan, and Vadhan~\cite{LRTV09} constructed a PRG for the special case of abelian groups with seed length
\[
    \widetilde{O}\left(\log(nw/\eps)\cdot \log w\right).
\]
More recently, for \(p\)-groups of constant size, Lee and Viola~\cite{LV25} obtained the optimal seed length
\(
O(\log(n/\eps))
\)
via a reduction to the construction of PRGs for polynomials over \(\FF_p\).

In a different line of work, more closely related to ours in its techniques, Pyne and Vadhan~\cite{PV21} constructed a weighted PRG for permutation BPs with seed length
\[
    \widetilde{O}\!\left(\log(1/\eps)+\sqrt{\log(n/\eps)}\cdot \log n\right).
\]
Their construction improves upon the known PRG constructions in the low-error regime. Hoza, Pyne, and Vadhan~\cite{HPV21} constructed a PRG for unbounded-width permutation BPs with a single accepting state, having seed length
\[
    O\!\left((\log\log n+\log(1/\eps))\cdot \log n\right).
\]
More recently, Cheng and Wu~\cite{CW26} obtained a weighted PRG for unbounded-width permutation BPs with seed length
\[
    O\!\left(\log(1/\eps)+(\log\log n+\sqrt{\log(1/\eps)})\cdot \log n\right).
\]

Hoza, Pyne, and Vadhan~\cite{HWP24} studied the limitations of the INW generator and proved that any ``spectral analysis'' of it requires seed length
\[
    \Omega\left((\log w+\log(1/\eps))\cdot \log n\right),
\]
even for permutation BPs. Here, a spectral analysis is one that assumes only a suitable spectral gap for the underlying expanders, without exploiting any additional properties. This framework encompasses all previous analyses of the INW generator.

In the more general setting of permutation BPs over an alphabet of size $d$, \cite{HWP24} proved a lower bound of
\(
\Omega(\log n \cdot \log\log \min\{n,d\}).
\)
In \cref{sec:lower-bound}, we tighten this lower bound to
\(\Omega(\log n \cdot \min\{\log d,\log\log n\})\). A direct generalization of our analysis yields a PRG with seed length
\[
    O((\log (1/\varepsilon) + \log w + \log d) \cdot \log n).
\]
Combined with the seed length obtained by \cite{HPV21}, this implies a bound that is tight up to constant factors within the spectral analysis framework.
We discuss this in more detail in \cref{sec:lower-bound}.

\subsection{Proof Idea}

A central ingredient in the BRRY analysis of the INW generator is the notion of the \emph{weight} of a branching program. This quantity measures, in an averaged sense, how much the individual layers can influence the final value of the computation. BRRY used this notion to sharpen the error analysis of INW: instead of accumulating error linearly over the $n$ layers of the program, the error is controlled in terms of the recursion depth, which is logarithmic in $n$, and the weight of the program. The second key step in their analysis is to prove that, for regular BPs, the total weight is bounded by $\poly(w)$ -- independent of the length of the program. Combining these two ingredients yields their improved seed length.

Our key idea is to combine the BRRY weight perspective with the reversibility available in permutation BPs. In addition to the usual BRRY weight, which we view as a \emph{backward} weight, we introduce a complementary \emph{forward} weight corresponding to the reversed program. This symmetric use of forward and backward weights is the main new ingredient in our analysis. It allows us to avoid paying for the  recursion depth altogether, thereby yielding the optimal dependence on $n$. 

Exploiting this symmetry requires the generator itself to preserve a compatible forward--backward structure. In our analysis, the seed-recycling step is not treated by an arbitrary sampler, but rather as the averaging operator of an expander, and can hence also recycle the seed when considering the inverse program. This feature is crucial for simultaneously controlling the forward and backward quantities that arise in the analysis. For this reason, our proof relies essentially on the expander-based structure of the INW construction. Interestingly, although for a seemingly different reason, a recent work by Hoza and Lv \cite{HL25} also relies heavily on the fact that the INW generator is based on expander graphs.


At a conceptual level, our proof can also be viewed as an error-propagation analysis tailored to the branching program at hand. Rather than controlling the error throughout the recursion using a single fixed norm, such as the spectral norm or an induced $\ell_1$ norm, as in most previous analyses, we associate with each subprogram its own seminorm, capturing the particular ``geometry'' of that subprogram. The main point is that, with respect to these adapted error measures, the error \emph{does not accumulate} over the recursion: it remains bounded by $O(\lambda)$, where $\lambda$ is the spectral expansion of the underlying expanders, provided $\lambda$ is taken to be sufficiently small, namely inverse-polynomial in $w$ and, crucially, independent of $n$.

While most prior work measures error using a fixed norm, several recent works instead consider program-dependent norms or seminorms \cite{MRS19,AKM20,PV21,HPV21,APP23,chen2023weighted} to give better error reduction procedures---either algorithmically in a white-box setting or in the construction of weighted PRGs---or for controlling the error propagation in INW-inspired constructions throughout the recursion. We discuss the relationship between the seminorm used in these works and our seminorm, in greater detail in \cref{sec:variance-seminorm-comparison}, focusing on \cite{chen2023weighted} (that builds on the singular-value approximation \cite{APP23}).


Beyond its conceptual appeal, we believe that the seminorm formulation may
be useful for combining our argument with other work. Many analyses of
PRGs and related recursive constructions are phrased in
terms of norm bounds, and the forward--backward seminorm isolates the precise
norm-like quantity that our proof controls.

\paragraph{Organization.}
After the brief preliminaries in \cref{sec:prelim}, we prove our main result, \cref{thm:informal-main}, in \cref{sec:main}. The proof makes no explicit reference to the seminorms discussed above. We believe that presenting the argument in this form makes it more direct and easier to follow. In \cref{sec:seminorm}, we reinterpret the proof from the perspective of the forward--backward seminorms. In \cref{sec:refined}, we give a refined analysis of the proof of \cref{thm:informal-main}. This refinement may be useful in future work and also brings out a convolutional perspective on the argument.
In \cref{sec:lower-bound} we tighten the lower bound on the seed length of spectral analyses of the INW generator.

\section{Preliminaries}\label{sec:prelim}

\subsection{Branching Programs}

We recall the relevant definitions of read-once branching programs, and PRGs that fool them. Those definitions are standard, and we refer the reader to \cite{Hoz24} for a more detailed discussion.

\begin{definition}[Branching programs (BPs)]\label{def:bp}
    Let $n,w$ be positive integers. A width-$w$ length-$n$ read-once branching program, denoted by $P$, is a directed graph on the set of vertices, called "states", which is arranged in layers $V = V_0 \cup \cdots \cup V_n$, each of size $w$.
    For every $t \in [n]$ and a vertex $v \in V_{t-1}$, there are two outgoing edges from $v$ to $V_{t}$, with labels 0 and 1. There is a ``start vertex'' $s \in V_0$, and a set of ``accept states'' $A \subseteq V_n$.
    A string $x \in \{0,1\}^n$ defines a walk through the program -- start at $s$, and at each step $t \in [n]$, walk through the outgoing edge which is labeled $x_t \in \{0,1\}$. This walk arrives at a final vertex $v \in V_n$ that is either ``accept'' ($v \in A$) or ``reject'' ($v \notin A$).
\end{definition}
It is useful to view each layer transition as a matrix. Identify each layer with $[w]$. For \(t\in[n]\) and \(b\in\{0,1\}\), let
\(P_{t,b}\in\{0,1\}^{w\times w}\) be the matrix whose \((i,j)\)-entry is \(1\)
iff the edge labeled \(b\) from the \(i\)-th vertex of \(V_{t-1}\) goes to the
\(j\)-th vertex of \(V_t\).
\begin{definition}[Permutation BP]
    A read-once branching program $P$ is called a permutation BP if, for all $t \in [n]$ and $b \in \{0,1\}$, the matrix $P_{t,b}$ is a permutation matrix.
\end{definition}
From now on, fix some width-$w$ length-$n$ BP $P$. 
We define
\[
T_t = \frac{P_{t,0}+P_{t,1}}{2},
\]
which is the random-walk operator associated with layer $t$.
We will frequently consider subprograms corresponding to intervals of the original BP. For an interval $I=\{a,\ldots,b-1\} \subseteq[n]$ and an input string $x\in\{0,1\}^{I}$, define
\(
    P_I(x)=P_{a,x_a}\cdots P_{b-1,x_{b-1}}.
\)
The corresponding uniform random-walk operator is
\begin{equation}\label{eq:real ti}
    T_I
    =\mathbb{E}_{x\in\{0,1\}^{I}}\!\left[ P_I(x) \right]
    =T_aT_{a+1}\cdots T_{b-1}.
\end{equation}

Let $A\subseteq V_n$ denote the set of accepting states in the $n$-th layer, and let $s\in V_0$ denote the starting state. The value of the program on input $x \in \{0,1\}^n$ is then given by
\(
    e_s^\top P_{[n]}(x) \mathbbm{1}_A,
\)
where $e_s\in\mathbb{R}^w$ is the standard basis vector corresponding to $s$, and $\mathbbm{1}_A\in\mathbb{R}^w$ is the indicator vector of $A$. For brevity, we may also denote this value by \(P(x)\). Similarly, the acceptance probability according to a truly uniform walk is given by $e_s^\top T_{[n]} \mathbbm{1}_A$.

A PRG for permutation BPs maps short uniformly random seeds to inputs while approximately preserving the acceptance probability of every program.
\begin{definition}
    A function $G:\{0,1\}^s \rightarrow \{0,1\}^n$ is an $\varepsilon$-PRG for width-$w$ length-$n$ permutation branching programs if for every such program $P$, we have
    \[
        \left | \mathbb{E}_{x \in \{0,1\}^n} [ P(x) ] - \mathbb{E}_{y \in \{0,1\}^s} [ P(G(y)) ] \right | \leq \varepsilon.
    \]
    In this case we say that $G$ $\varepsilon$-fools length-$n$ width-$w$ permutation BPs.
\end{definition}

\subsection{Expander Graphs}
As mentioned, the PRG we use is the INW generator, which is constructed from a family of expander graphs. We begin by recalling the definition of an expander.
\begin{definition}[Expander graph]
    Let $G = (V,E)$ be an undirected, $d$-regular graph, and let $A$ be its adjacency matrix. We say that the graph $G$ is a $\lambda$-expander if the spectral norm of the matrix $\frac1{d} \,A - J$ is at most $\lambda$, where $J$ is the normalized all-ones matrix.
\end{definition}
For the INW generator, we need the following construction of an explicit family of expander graphs.

\begin{theorem}[e.g., \cite{MRS19}, Theorem 3.3]\label{thm:expander we use}
    For every $n \in \mathbb{N}$ and $\lambda > 0$, there exists an explicit graph $H_n$ on $n$ vertices, which is a $D$-regular $\lambda$-expander with $D = (1/\lambda)^{\Theta(1)}$. Moreover, for any vertex $x$ and a neighbor label $y \in [D]$, the $y$-th neighbor of $x$, which is denoted by $H_n(x,y)$, can be computed in space $O(\log (nD))$.
\end{theorem}

We will use the following vector-valued version of the expander mixing lemma. 
\begin{lemma}\label{lemma:expander-mixing-vectors}
    Let $H = (V,E)$ be an expander graph with second eigenvalue $\lambda \in (0,1)$. Let $\alpha,\beta:V \to \mathbb{R}^k$ be two functions such that $\mathbb{E}_{x \in V} [\alpha(x)] = \mathbb{E}_{y \in V} [\beta(y)] = 0$. Then,
    \[
        \left |\mathbb{E}_{(x,y) \in E} [\langle \alpha(x), \beta(y) \rangle] \right | \leq \lambda \left ( \mathbb{E}_{x \in V} [\lVert\alpha(x)\rVert_2^2] \right )^{1/2} \left ( \mathbb{E}_{y \in V} [\lVert\beta(y)\rVert_2^2] \right )^{1/2}.
    \]
\end{lemma}

The proof follows from the standard expander-mixing lemma; see \cref{sec:missing proofs} for the details.

\subsection{The INW Generator}

We now recall the INW construction.
Fix parameters \(n,w,\varepsilon\). We assume without loss of generality that \(n = 2^h\), since otherwise the program can be padded with trivial layers.
Let $d$ be an integer parameter to be chosen later. For every $i=1,2,\ldots,h$, let $H_i$ be a $2^d$-regular $\lambda$-expander on the vertex set $\{0,1\}^{(i-1)d+1}$, given by \cref{thm:expander we use}.

The INW generator is defined recursively as follows. For $i\geq 0$, we define a function
\[
    G_i \colon \{0,1\}^{id+1}\to\{0,1\}^{2^i}.
\]
For $G_0 \colon \{0,1\}\to\{0,1\}$, we set $G_0(x)=x$. For $i\geq 1$, write an input to $G_i$ as $(x,y)$, where
\[
    x\in\{0,1\}^{(i-1)d+1}
    \qquad\text{and}\qquad
    y\in\{0,1\}^d.
\]
Then,
\[
    G_i(x,y)
    =
    \bigl(
        G_{i-1}(x),
        G_{i-1}(H_i(x,y))
    \bigr).
\]

The generator for length-$n$ programs is therefore $G_h$, where recall that $h=\log_2 n$. Its seed length is $O(d\log n)$. We instantiate the construction with the expanders given by \cref{thm:expander we use}. Consequently, the overall seed length, as well as the space complexity of the generator, is
$O\!\left(\log(1/\lambda)\cdot \log n\right)$.

We define the analogue of $T_I$ (see \cref{eq:real ti}) for an interval $I\subseteq[n]$, with respect to the pseudorandom walk induced by the INW generator. More precisely, for $i\geq 0$ and for an interval $I\subseteq[n]$ of size $2^i$, define
\[
    \wt{T}_I
    =
    \EE_x\left[ P_I(G_i(x)) \right],
\]
where the expectation is over a uniformly random seed $x$ for $G_i$.

\section{Proof of Main Result}\label{sec:main}

 In this section we prove the following theorem, which is the formal and complete statement of \cref{thm:informal-main}.
 
 \begin{theorem}\label{thm:main-formal-result}
    For every $n,w \in \mathbb{N}$ and for all $0 < \varepsilon \leq 1$ the expander-based INW-generator set with $\lambda = \frac{\varepsilon}{32 w^3}$ $\varepsilon$-fools length-$n$ width-$w$ permutation BPs. It has seed length $$O(\log (1/\lambda) \cdot \log n) = O((\log (1/\varepsilon) + \log w) \cdot \log n).$$
\end{theorem}

The proof is organized as follows. In \cref{subsec:fb-weight}, we introduce the forward and backward weights and establish two useful properties that will be used later. In \cref{subsec:bound-expander-product}, we bound the error of the expander product in terms of these weights; this is the key estimate underlying our error analysis. Finally, in \cref{subsec:error-analysis}, we analyze the error across the recursion and show that it does not deteriorate from one level to the next, thereby proving \cref{thm:main-formal-result}.

\subsection{Forward and Backward Weights}\label{subsec:fb-weight}

The main ingredient in our error analysis is the notion of the program's weight, following \cite{BRRY}. Our key insight is that, for permutation BPs, the weight of the \emph{reversed} program can also be used in the bound. In this section, we define the notions of forward weight and backward weight, where the latter corresponds to the BRRY notion of weight, expressed in the language of operators. Throughout this section, all definitions are with respect to a fixed permutation branching program $P$.

\paragraph{Forward weights.}
Let $1\leq a\leq b\leq n+1$ be integers, let $I=\{a,a+1,\ldots,b-1\}$, and let $p=p_a\in\mathbb{R}^w$. The \emph{forward propagation of $p_a$ along the interval $I$} is the sequence of vectors
$
    p_a,p_{a+1},\ldots,p_b
$
defined recursively, for $t = a,a+1,\ldots,b-1$, by
\[
    p_{t+1}^\top
    =
    p_t^\top T_t
    =
    p_t^\top \cdot \frac{P_{t,0}+P_{t,1}}{2}.
\]
When $p_a$ is a probability distribution, this is precisely the random walk induced by the program, restricted to the interval $I$.

The \emph{forward weight at layer $t\in I$} is defined by
\[
    F_t(p)
    =
    \norm{p_t^\top P_{t,0}-p_{t+1}^\top}_1
    +
    \norm{p_t^\top P_{t,1}-p_{t+1}^\top}_1
    =
    \norm{p_t^\top(P_{t,0}-P_{t,1})}_1.
\]
The \emph{forward weight over the interval $I$} is
\[
    F_I(p)
    =
    \sum_{t=a}^{b-1} F_t(p).
\]
Observe that if we write the interval $I$ as the concatenation of two consecutive intervals $I=L\circ R$, then
\begin{equation}\label{eq:forwards concat}
    F_I(p)
    =
    F_L(p)
    +
    F_R(T_L^\top p).
\end{equation}

\paragraph{Backward weights.}
Let $1\leq a\leq b\leq n+1$ be integers, let $I=\{a,\ldots,b-1\}$, and let $q=q_b\in\mathbb{R}^w$. The \emph{backward propagation of $q_b$ along the interval $I$} is the sequence of vectors
$
    q_b,q_{b-1},\ldots,q_a
$
defined recursively, for $t=b-1,b-2,\ldots,a$, by
\[
    q_t
    =
    T_t q_{t+1}
    =
    \frac{P_{t,0}+P_{t,1}}{2}\cdot q_{t+1}.
\]
If $q_b$ is a probability distribution, then this is the random walk corresponding to the reversed program, since
\[
    q_t^\top
    =
    q_{t+1}^\top T_t^\top
    =
    q_{t+1}^\top\cdot \frac{P_{t,0}^\top+P_{t,1}^\top}{2}.
\]
The \emph{backward weight at layer $t\in I$} is defined by
\[
    B_t(q)
    =
    \norm{P_{t,0}q_{t+1}-q_t}_1
    +
    \norm{P_{t,1}q_{t+1}-q_t}_1
    =
    \norm{(P_{t,0}-P_{t,1})q_{t+1}}_1.
\]
The \emph{backward weight over the interval $I$} is
\[
    B_I(q)
    =
    \sum_{t=a}^{b-1} B_t(q).
\]
Similarly to \cref{eq:forwards concat}, if we write the interval $I$ as the concatenation of two consecutive intervals $I=L\circ R$, then
\begin{equation}\label{eq:backward concat}
    B_I(q)
    =
    B_L(T_R q)
    +
    B_R(q).
\end{equation}

\begin{lemma}\label{lemma:weight-bound}
    Let $P$ be a permutation branching program of width $w$. For every interval
    $I=\{a,\ldots,b-1\}$ and every $p,q\in\mathbb{R}^w$, we have
    \begin{align*}
        B_I(q)
        &\leq
        2\sum_{1\leq i<j\leq w}|q_i-q_j|
        \leq
        2w\norm{q}_1, \\
        F_I(p)
        &\leq
        2\sum_{1\leq i<j\leq w}|p_i-p_j|
        \leq
        2w\norm{p}_1.
    \end{align*}
\end{lemma}

The proof follows from Lemma~5 of \cite{BRRY}, applied both to $P$ and to the reversed program. We also use the observation that the proof of Lemma~5 applies verbatim to arbitrary vectors in $\mathbb{R}^w$; it is not essential that the vector entries are in $[0,1]$. See \cref{sec:missing proofs}.

\begin{lemma}\label{lemma:random-vs-deterministic-walk}
    Fix some permutation branching program $P$ of length $n$ and width $w$, and let $1 \le a \le b \le n+1$ be integers.
    Let $I=\{a,\ldots,b-1\}$, and let $p,q\in\mathbb{R}^w$. Then, for every
    $x\in\{0,1\}^I$, we have
    \begin{align*}
        \norm{P_I(x)^\top p - T_I^\top p}_1
        &\leq
        \frac{1}{2}F_I(p), \\
        \norm{P_I(x)q - T_Iq}_1
        &\leq
        \frac{1}{2}B_I(q).
    \end{align*}
\end{lemma}

\begin{proof}
    We first prove the forward inequality. Let
    $
        p_a,p_{a+1},\ldots,p_b
    $
    be the forward propagation of $p_a=p$ along $I$, so that
    $p_{t+1}=T_t^\top p_t$. For $t \in \{a,a+1,\ldots,b\}$, define
    \[
        r_t
        =
        \left(P_{a,x_a}P_{a+1,x_{a+1}}\cdots P_{t-1,x_{t-1}}\right)^\top p,
    \]
    with the convention that $r_a=p$. Thus $r_b=P_I(x)^\top p$ and
    $p_b=T_I^\top p$. Moreover,
    $
        r_{t+1}
        =
        P_{t,x_t}^\top r_t .
    $
    Therefore, for every $t\in I$,
    \begin{align*}
        \norm{r_{t+1}-p_{t+1}}_1
        &=
        \norm{P_{t,x_t}^\top r_t - p_{t+1}}_1 \\
        &\leq
        \norm{P_{t,x_t}^\top(r_t-p_t)}_1
        +
        \norm{P_{t,x_t}^\top p_t-p_{t+1}}_1 \\
        &=
        \norm{r_t-p_t}_1
        +
        \norm{P_{t,x_t}^\top p_t-p_{t+1}}_1 \\
        &=
        \norm{r_t-p_t}_1
        +
        \frac{1}{2} F_t(p),
    \end{align*}
    where we used that permutation matrices preserve the $\ell_1$ norm. Since
    $r_a=p_a=p$, induction gives
    \[
        \norm{P_I(x)^\top p - T_I^\top p}_1
        =
        \norm{r_b-p_b}_1
        \leq
        \frac{1}{2}\sum_{t=a}^{b-1} F_t(p)
        =
        \frac{1}{2}F_I(p).
    \]
    
    The claim for the backward inequality is proved in a similar way.
\end{proof}

\subsection{Analyzing the Expander Product}\label{subsec:bound-expander-product}

In this subsection, we prove the following proposition, which is the key estimate in our error analysis.

\begin{proposition}\label{prop:local-expander-product-error}
    Fix some permutation branching program of length $n$ and width $w$.
    Let $p,q\in\mathbb{R}^w$, and let $I=L\circ R\subseteq[n]$ be an interval of size $2^i$, where $L$ and $R$ are its left and right halves. Then
    \[
        \left|
            p^\top\bigl(\wt{T}_I-\wt{T}_L\wt{T}_R\bigr)q
        \right|
        \leq
        \lambda\cdot F_L(p)\cdot B_R(q).
    \]
\end{proposition}

\begin{proof}
    Let $\Omega=\{0,1\}^{(i-1)d+1}$
    be the seed space of $G_{i-1}$. For $x,z\in\Omega$, define
    \[
        A_x=P_L(G_{i-1}(x)),
        \qquad
        B_z=P_R(G_{i-1}(z)).
    \]
    Then
    \[
        \wt{T}_L=\EE_x A_x,
        \qquad
        \wt{T}_R=\EE_z B_z,
        \qquad
        \wt{T}_L\wt{T}_R=\EE_{x,z} A_xB_z,
    \]
    where $x,z$ are independent and uniform over $\Omega$. Moreover,
    \[
        \wt{T}_I=\EE_{(x,z)\sim H_i} A_xB_z,
   \]
    where $(x,z)\sim H_i$ means that $x$ is uniform over $\Omega$ and $z$ is a uniformly random neighbor of $x$ in $H_i$.     Define
    \[
        \alpha_x=A_x^\top p-\wt{T}_L^\top p,
        \qquad
        \beta_z=B_zq-\wt{T}_Rq.
    \]
    Then, $\EE_x\alpha_x=0$ and $\EE_z\beta_z=0$. Moreover,
    \begin{align*}
        \langle \alpha_x,\beta_z\rangle
        &=
        p^\top(A_x-\wt{T}_L)(B_z-\wt{T}_R)q \\
        &=
        p^\top
        \bigl(
            A_xB_z-\wt{T}_LB_z-A_x\wt{T}_R+\wt{T}_L\wt{T}_R
        \bigr)q.
    \end{align*}
    Taking expectation over $(x,z)\sim H_i$, and using that both marginals of this distribution are uniform over $\Omega$, we get
    \begin{equation}\label{eq:coffee}
        \EE_{(x,z)\sim H_i}\langle \alpha_x,\beta_z\rangle
        =
        p^\top(\wt{T}_I-\wt{T}_L\wt{T}_R)q.
    \end{equation}
    Hence, by \cref{lemma:expander-mixing-vectors},
    \[
        \left|
        p^\top(\wt{T}_I-\wt{T}_L\wt{T}_R)q
        \right|
        \leq
        \lambda
        \left(\EE_x\norm{\alpha_x}_2^2\right)^{1/2}
        \left(\EE_z\norm{\beta_z}_2^2\right)^{1/2}.
    \]

    It remains to bound the two terms on the right-hand side. For every $x\in\Omega$, 
    \[
        \alpha_x
        =
        A_x^\top p-T_L^\top p
        +
        T_L^\top p-\wt{T}_L^\top p.
    \]
    By \cref{lemma:random-vs-deterministic-walk},
    \[
        \norm{A_x^\top p-T_L^\top p}_1\leq \frac{1}{2} F_L(p).
    \]
    Also,
    \[
        T_L^\top p-\wt{T}_L^\top p
        =
        \EE_{x'}\left[T_L^\top p-A_{x'}^\top p\right],
    \]
    and therefore, again by \cref{lemma:random-vs-deterministic-walk} and convexity,
    \[
        \norm{T_L^\top p-\wt{T}_L^\top p}_1\leq \frac{1}{2} F_L(p).
    \]
    Thus
    \[
        \norm{\alpha_x}_2\leq \norm{\alpha_x}_1\leq F_L(p).
    \]
    Similarly, for every $z\in\Omega$,
    \[
        \norm{\beta_z}_2\leq \norm{\beta_z}_1\leq B_R(q).
    \]
    Substituting these bounds gives
    \[
        \left|
        p^\top(\wt{T}_I-\wt{T}_L\wt{T}_R)q
        \right|
        \leq
        \lambda F_L(p)B_R(q),
    \]
    as required.
\end{proof}

\subsection{Error Analysis}\label{subsec:error-analysis}

In this section we analyze the error propagation throughout the recursion.  

\begin{proposition}\label{prop:error-bound}
    Let $\lambda > 0$ be a bound on the second eigenvalues of the expander graphs in the expander-based INW construction.
    Fix a permutation branching program $P$ of length $n$ and width $w$.
    Let $c\geq 2$, and assume that
    \[
        \lambda \leq \frac{1}{8c^2w^3}.
    \]
    Let $p,q\in\mathbb{R}^w$, and let $I\subseteq[n]$ be an interval of size $2^i$. Define
    \[
        \EEE_I(p,q)
        =
        \left|p^\top(\wt{T}_I-T_I)q\right|.
    \]
    Then,
    \[
        \EEE_I(p,q)
        \leq
        c\lambda F_I(p)B_I(q).
    \]
\end{proposition}

\begin{proof}
    We prove the proposition by induction on $i$. When $i=0$, we have
    $\wt{T}_I=T_I$, and the claim is immediate. Assume therefore that $i\geq 1$,
    and that the claim holds for intervals of size $2^{i-1}$.

    Write $I=L\circ R$, where $L$ and $R$ are the left and right halves of $I$.
    Set
    \[
        A=F_L(p),
        \qquad
        B=F_R(T_L^\top p),
        \qquad
        C=B_L(T_Rq),
        \qquad
        D=B_R(q).
    \]
    By the decomposition properties of forward and backward weights (see \cref{eq:forwards concat,eq:backward concat})
    \[
        F_I(p)=A+B,
        \qquad
        B_I(q)=C+D.
    \]
    Thus it suffices to prove
    \[
        \EEE_I(p,q)
        \leq
        c\lambda(A+B)(C+D).
    \]

    We decompose the error as
    \begin{align}
        p^\top(\wt{T}_I-T_I)q
        &=
        p^\top(\wt{T}_I-\wt{T}_L\wt{T}_R)q
        \label{summand:expander-prod} \\
        &\quad+
        p^\top\wt{T}_L(\wt{T}_R-T_R)q
        \label{summand:right-error} \\
        &\quad+
        p^\top(\wt{T}_L-T_L)T_Rq.
        \label{summand:left-error}
    \end{align}
    We bound the three terms separately.

    For \eqref{summand:expander-prod}, \cref{prop:local-expander-product-error} gives
    \[
        \left|
            p^\top(\wt{T}_I-\wt{T}_L\wt{T}_R)q
        \right|
        \leq
        \lambda AD.
    \]
    For \eqref{summand:left-error}, the induction hypothesis applied to the
    interval $L$, with vectors $p$ and $T_Rq$, gives
    \[
        \left|
            p^\top(\wt{T}_L-T_L)T_Rq
        \right|
        \leq
        c\lambda AC.
    \]
    For \eqref{summand:right-error}, the induction hypothesis applied to the
    interval $R$, with vectors $\wt{T}_L^\top p$ and $q$, gives
    \begin{equation}\label{eq:milo}
        \left|
            p^\top\wt{T}_L(\wt{T}_R-T_R)q
        \right|
        =
        \left|
            (\wt{T}_L^\top p)^\top(\wt{T}_R-T_R)q
        \right|
        \leq
        c\lambda D\cdot F_R(\wt{T}_L^\top p).
    \end{equation}
    We next compare $F_R(\wt{T}_L^\top p)$ with
    $F_R(T_L^\top p)=B$.

    Let
    \[
        \delta=\wt{T}_L^\top p-T_L^\top p.
    \]
    Since both $T_L$ and $\wt{T}_L$ are doubly stochastic, we have
    $\sum_{j=1}^w\delta_j=0$. Let $r\in\{0,1\}^w$ be the indicator vector of
    the set of coordinates on which $\delta$ is nonnegative, namely
    \[
        r_j=
        \begin{cases}
            1, & \delta_j\geq 0, \\
            0, & \text{otherwise}.
        \end{cases}
    \]
    Then
    \[
        \norm{\delta}_1
        =
        2\langle \delta,r\rangle.
    \]
    Moreover,
    \[
        \langle \delta,r\rangle
        =
        p^\top(\wt{T}_L-T_L)r.
    \]
    Applying the induction hypothesis to the interval $L$, with vectors $p$ and
    $r$, we obtain
    \[
        |\langle \delta,r\rangle|
        \leq
        c\lambda F_L(p)B_L(r).
    \]
    By \cref{lemma:weight-bound}, and since $\norm{r}_1\leq w$,
    \[
        B_L(r)\leq 2w\norm{r}_1\leq 2w^2.
    \]
    Hence
    \[
        \norm{\delta}_1
        \leq
        4w^2c\lambda A.
    \]

    Using the subadditivity of $F_R(\cdot)$ and \cref{lemma:weight-bound}, we get
    \[
        F_R(\wt{T}_L^\top p)
        \leq
        F_R(T_L^\top p)+F_R(\delta)
        \leq
        B+2w\norm{\delta}_1
        \leq
        B+8w^3c\lambda A.
    \]
    Since $\lambda\leq 1/(8c^2w^3)$, it follows that
    \[
        F_R(\wt{T}_L^\top p)
        \leq
        B+\frac{A}{c}.
    \]
    Therefore,
    \[
        \left|
            p^\top\wt{T}_L(\wt{T}_R-T_R)q
        \right|
        \leq
        c\lambda DB+\lambda DA.
    \]

    Combining the three bounds, we conclude that
    \begin{equation}\label{eq:before}
        \EEE_I(p,q)
        \leq
        2\lambda AD+c\lambda AC+c\lambda DB.
    \end{equation}
    Since $c\geq 2$, this is at most
    \begin{equation}\label{eq:after}
        c\lambda(A+B)(C+D),
    \end{equation}
    completing the induction.
\end{proof}

We are finally ready to prove our main result.

\begin{proof}[Proof of \cref{thm:main-formal-result}]
    Without loss of generality we may assume that $n$ is a power of $2$, since otherwise we may add at most $n$ identity layers to the program.    
    By applying \cref{prop:error-bound} and \cref{lemma:weight-bound} with $c=2$, $q = \mathbbm{1}_A$, $p = e_s$, we conclude
    \[
        \left |p^{\top} \left ( T_{[n]} - \wt{T}_{[n]} \right ) q \right | \leq c \lambda F_{[n]}(p)B_{[n]}(q) \leq 4 c \lambda w^2 |A| < \varepsilon.
    \]
\end{proof}

\subsection{BPs with Larger Alphabet}

In this short subsection, we briefly present the generalization of our construction for permutation BPs over an alphabet of size $d \geq 2$ -- that is, a BP in which each vertex $v$ has $d$ outgoing edges from $v$ to the next layer, with labels in $[d]$. Naturally, a walk through the program is defined by a string $x \in [d]^n$.

An analogous construction of the INW generator for this setting admits a similar analysis.
Use the same INW recursion, starting from $G_0(x)=x, x\in[d]$, and retain the recursive seed-recycling step.
For the analysis, set
\[
T_t=\frac{1}{d}\sum_{b\in[d]}P_{t,b},
\]
and define
\[
F_t(p)=\sum_{b\in[d]}\|p_t^\top P_{t,b}-p_{t+1}^\top\|_1,
\qquad
B_t(q)=\sum_{b\in[d]}\|P_{t,b}q_{t+1}-q_t\|_1.
\]
The bound in \cref{lemma:weight-bound} incurs an extra factor of $d-1$.
\cref{lemma:random-vs-deterministic-walk} and \cref{prop:local-expander-product-error} remain valid, and the two applications of \cref{lemma:weight-bound} in \cref{prop:error-bound} replace its condition by
\[
\lambda\le \frac{1}{8c^2(d-1)^2w^3}.
\]

This implies the following theorem.

\begin{theorem}\label{thm:general-alphabet-upper}\label{thm:our-with-general-d}
    For every $n,d,w \geq 2$ and for all $0 < \varepsilon \leq 1$ the expander-based INW-generator set with $\lambda = \frac{\varepsilon}{32 d^2 w^3}$ $\varepsilon$-fools length-$n$ width-$w$ permutation BPs over $[d]$. It has seed length
    $$
        O(\log (1/\lambda) \cdot \log n) = O((\log (1/\varepsilon) + \log w + \log d) \cdot \log n).
    $$
\end{theorem}

In \cref{sec:lower-bound}, we prove that this bound, together with that of \cite{HPV21}, tightly matches the optimal seed length attainable via any spectral analysis of the INW generator for permutation BPs.

\section{The Forward--Backward Seminorm}\label{sec:seminorm}

In this short section we present a different perspective on the proof of
\cref{prop:local-expander-product-error}. Namely, we view the proof as controlling the propagation of
error through the recursion in terms of a family of seminorms adapted to the
underlying subprograms.

\begin{definition}
Fix a permutation branching program of length $n$ and width $w$, and let
$I \subseteq [n]$ be an interval. For a $w \times w$ nonzero real matrix $\Delta$, define
the forward--backward seminorm by\footnote{If the supremum in the definition is infinite or taken over an empty set, we define the seminorm to be $\infty$.}
\[
        \|\Delta\|_{\mathrm{fb},I}
        =
        \sup_{p,q:\,F_I(p)B_I(q)>0}
        \frac{|p^\top \Delta q|}{F_I(p)B_I(q)}.
\]
For the zero matrix, set $\|0\|_{\mathrm{fb},I}=0$.
\end{definition}

Thus, every interval $I$ of the program induces its own seminorm, tailored to
the behavior of the branching program on that interval. In other words, the
error is measured with respect to the ``geometry'' induced by the
corresponding subprogram.

With this terminology, \cref{prop:error-bound} can be restated in the following form. 

\begin{proposition}\label{prop:error-bound-seminorm}
Fix a permutation
branching program of length $n$ and width $w$. Let $c \ge 2$, and assume that
\(
        \lambda \le \frac{1}{8c^2 w^3}.
\)
Then, for every dyadic interval $I \subseteq [n]$,
\[
        \|T_I - \widetilde{T}_I\|_{\mathrm{fb},I} \le c\lambda .
\]
\end{proposition}

This formulation highlights the conceptual content of the proof of
\cref{prop:error-bound}. Rather than controlling error propagation throughout
the recursion using a single fixed norm, such as the spectral norm or the
induced $\ell_1$ norm, the proof uses a different seminorm for each dyadic
subprogram. These seminorms are chosen to match the forward and
backward structure of the corresponding interval. The main point is that,
with respect to these adapted error measures, the error does not deteriorate
through the recursion beyond a universal bound of $O(\lambda)$ provided $\lambda$ is taken sufficiently small.

\subsection{Relation to the Seminorm of \texorpdfstring{\cite{chen2023weighted}}{Chen et al.}}
\label{sec:variance-seminorm-comparison}

In this section, we discuss how our seminorm relates to the one used by Chen, Hoza, Lyu, Tal, and Wu~\cite{chen2023weighted} that builds on the singular-value approximation \cite{APP23}.
 Although the two seminorms are different and arise in  different contexts---theirs is used for the non-black-box derandomization of regular ROBPs, whereas ours is used to control error throughout the INW recursion---we believe that the comparison is instructive.

Let $I=L\circ R$ correspond to one step of the INW recursion, and let
$\Omega$ be the seed space of the child generator. Recall the notation
introduced in the proof of \cref{prop:local-expander-product-error}. For
$x,z\in\Omega$, let $A_x=P_L(G_{i-1}(x))$, $
B_z=P_R(G_{i-1}(z))$, 
and recall that
$\widetilde T_L=\mathbb E_x[A_x]$, $   \widetilde T_R=\mathbb E_z[B_z]$.
For $p,q\in\mathbb R^w$, we defined
$\alpha_x =A_x^{\top}p-\widetilde T_L^{\top}p$, $\beta_z =B_zq-\widetilde T_Rq$.
In \cref{eq:coffee}, we showed that
\[
    \mathbb E_{(x,z)\sim H_i}\langle\alpha_x,\beta_z\rangle
    =p^{\top}(\widetilde T_I-\widetilde T_L\widetilde T_R)q.
\]
Invoking \cref{lemma:expander-mixing-vectors}, we then deduced that
\begin{equation}
    \left|p^{\top}(\widetilde T_I-\widetilde T_L\widetilde T_R)q\right|
    \leq \lambda
    \bigl(\mathbb E_x\|\alpha_x\|_2^2\bigr)^{1/2}
    \bigl(\mathbb E_z\|\beta_z\|_2^2\bigr)^{1/2}.
    \label{eq:variance-mixing}
\end{equation}
At this point, our discussion departs from the proof of
\cref{prop:local-expander-product-error}.

Since $A_x$ and $B_z$ are permutation matrices, a direct calculation
shows that the two variances appearing on the right-hand side of
\cref{eq:variance-mixing} are exactly
\begin{align*}
    \mathbb E_x\|\alpha_x\|_2^2
    &=p^{\top}(I-\widetilde T_L\widetilde T_L^{\top})p,\\
    \mathbb E_z\|\beta_z\|_2^2
    &=q^{\top}(I-\widetilde T_R^{\top}\widetilde T_R)q.
\end{align*}
For a matrix $M$, define its left and right ``defect''\footnote{We refer to these as the left and right defect quantities of $M$, since
they measure the failure of $M^\top$ and $M$, respectively, to preserve
Euclidean norm:
$
    D_M^-(p)^2=\|p\|_2^2-\|M^\top p\|_2^2,\,
    D_M^+(q)^2=\|q\|_2^2-\|Mq\|_2^2$.
} quantities by
\[
    D_M^-(p):=\bigl(p^{\top}(I-MM^{\top})p\bigr)^{1/2},
    \qquad
    D_M^+(q):=\bigl(q^{\top}(I-M^{\top}M)q\bigr)^{1/2}.
\]
We can therefore rewrite \cref{eq:variance-mixing} as
\begin{equation}
    \left|p^{\top}(\widetilde T_I-\widetilde T_L\widetilde T_R)q\right|
    \leq
    \lambda D_{\widetilde T_L}^-(p)
            D_{\widetilde T_R}^+(q).
    \label{eq:local-defect-estimate}
\end{equation}
The forward and backward weights used in the proof of
\cref{prop:local-expander-product-error} yield upper bounds on these defect quantities. Although these bounds are coarser, they are more convenient for our purposes.

The quantities $D^-$ and $D^+$ also underlie the program-dependent
seminorm introduced in
\cite{chen2023weighted}. In the notation of their paper, for a walk
matrix $W_{r\leftarrow\ell}$ they define the interval seminorm
\begin{equation}
    \|u\|_{\ell\leadsto r}^2
    :=u^{\top}
      \bigl(I-W_{r\leftarrow\ell}^{\top}
               W_{r\leftarrow\ell}\bigr)u
    =\|u\|_2^2-\|W_{r\leftarrow\ell}u\|_2^2.
    \label{eq:chen-interval-defect}
\end{equation}
Under our transpose convention,
$W_{r\leftarrow\ell}=T_{[\ell+1,r+1)}^{\top}$, and hence
\[
    \|u\|_{\ell\leadsto r}^2
    =D_{T_{[\ell,r)}}^-(u)^2.
\]
Their seminorm, which they call the $\mathsf{F}$-seminorm---not to be confused
with our notation $F$ for the forward weight---aggregates these
quantities over the dyadic intervals of the layered program:
\begin{equation}
    \|x\|_{\mathsf F}^2
    =\sum_{k=1}^{\log n}\sum_{j\in U_k}
      \bigl\|x^{[j]}\bigr\|_{j\leadsto j+2^k}^2,
    \label{eq:chen-F-seminorm}
\end{equation}
where $U_k=\{0,2^k,2\cdot 2^k,\ldots,n\}$, with the convention adopted
in their paper for the terminal layer. In Section~9 of their paper, the
same defect quantities appear in the definition of singular-value
approximation:
\begin{equation}
    |y^{\top}(\widetilde W-W)x|
    \leq \frac{\lambda}{4}
    \left(D_W^+(x)^2+D_W^-(y)^2\right).
    \label{eq:chen-sv}
\end{equation}

Thus, although the two seminorms are not identical, they are based on
the same program-dependent notion of Euclidean dissipation. The
seminorm of \cite{chen2023weighted} is a one-sided, multiscale
$\ell_2$ aggregation over the entire layered space. By contrast, our
seminorm is associated with a single interval and measures a two-sided
bilinear error using the forward and backward $\ell_1$ weights.

\section{A Refined Analysis and a Convolution Form of the Error}\label{sec:refined}
In this section we present a more refined analysis of \cref{prop:error-bound}. The attentive reader may have noticed two sources of slack in the proof. The first is the mixture of $\ell_1$ and $\ell_2$ bounds. Indeed, the BRRY weight bound is stated in terms of the $\ell_1$ norm, which in particular implies the corresponding $\ell_2$ bound. By contrast, the proof of \cref{prop:error-bound} uses spectral estimates throughout, except for the estimate on $\delta$, which is given in terms of the $\ell_1$ norm. The second source of slack is that, in the induction step from \cref{eq:before} to \cref{eq:after}, the term $BC$ is in fact redundant. It turns out that both sources of slack can be removed, leading to the more refined analysis presented in this section. While this improvement does not affect our main result, \cref{thm:informal-main}, we believe that this perspective is worth recording with an eye toward future work.

Before turning to the formal proof, we illustrate the idea behind the refined
analysis. In \cref{prop:error-bound}, the error term
\[
        p^\top \widetilde T_L(\widetilde T_R-T_R)q
\]
was bounded by applying the induction hypothesis to the interval \(R\), with
\(\widetilde T_L^\top p\) as the left test vector. This required comparing \(F_R(\widetilde T_L^\top p)\) with \(F_R(T_L^\top p)\), the task we turned to after \cref{eq:milo}.
 It is better to compare
\(T_L^\top p\) with its approximation \(\widetilde T_L^\top p\) before passing them through the
functional \(F_R\). This can be done by introducing one additional hybrid, in
addition to the three summands appearing in
\cref{summand:expander-prod,summand:right-error,summand:left-error}:
\[
        \widetilde T_L(\widetilde T_R-T_R)
        =
        T_L(\widetilde T_R-T_R)
        +
        (\widetilde T_L-T_L)(\widetilde T_R-T_R).
\]
This introduces the second-order error term $(\widetilde T_L-T_L)(\widetilde T_R-T_R)$, but this term can be absorbed using
Cauchy--Schwarz together with the induction hypothesis. With this in mind, we define the more accurate potential function as
follows.

\begin{definition}[Recursive forward--backward potential]
Let \(I\) be a dyadic interval.  We define a quantity
\(\Phi_I(p,q)\) recursively as follows.  If \(|I|=1\), set
\[
        \Phi_I(p,q)=0.
\]
If \(I=L\circ R\), where \(L\) and \(R\) are the left and right halves of
\(I\), set
\[
        \Phi_I(p,q)
        =
        F_L(p)B_R(q)
        +
        \Phi_L(p,T_Rq)
        +
        \Phi_R(T_L^\top p,q).
\]
\end{definition}

The potential \(\Phi_I(p,q)\) is always upper bounded by the product
\(F_I(p)B_I(q)\), which is the quantity used in \cref{prop:error-bound}.

\begin{claim}
\label{lem:phi-dominated-by-fb}
For every dyadic interval \(I\) and every \(p,q\in\mathbb R^w\),
\[
        \Phi_I(p,q)\le F_I(p)B_I(q).
\]
\end{claim}

\begin{proof}
The claim is trivial when \(|I|=1\).  Suppose \(I=L\circ R\).  Define the parameters $A,B,C$ and $D$ as in the proof of \cref{prop:error-bound},
\[
        A=F_L(p),
        \qquad
        B=F_R(T_L^\top p),
        \qquad
        C=B_L(T_Rq),
        \qquad
        D=B_R(q).
\]
Then
\[
        F_I(p)=A+B,
        \qquad
        B_I(q)=C+D.
\]
By the induction hypothesis,
\[
        \Phi_L(p,T_Rq)\le AC,
        \qquad
        \Phi_R(T_L^\top p,q)\le BD.
\]
Therefore,
\[
\begin{aligned}
        \Phi_I(p,q)
        &=
        AD+\Phi_L(p,T_Rq)+\Phi_R(T_L^\top p,q)  \\
        &\le
        AD+AC+BD                                \\
        &\le
        (A+B)(C+D)                               \\
        &=
        F_I(p)B_I(q).
\end{aligned}
\]
This completes the induction.
\end{proof}

We will also use the following notation.  For a dyadic interval \(J\), set
\[
        \kappa_F(J)
        =
        \sup_{\|u\|_2=1} F_J(u),
        \qquad
        \kappa_B(J)
        =
        \sup_{\|v\|_2=1} B_J(v).
\]
Let
\[
        K
        =
        \sup_{J=L\circ R}
        \kappa_B(L)\kappa_F(R),
\]
where the supremum is over all dyadic intervals \(J\) appearing in the
recursion.

\begin{proposition}[Refined forward--backward induction]
\label{prop:refined-forward-backward}
Let \(c\ge 2\), and assume that
\[
        1+c^2\lambda K\le c.
\]
Then, for every dyadic interval \(I\) and every \(p,q\in\mathbb R^w\),
\[
        \left|
        p^\top(\widetilde T_I-T_I)q
        \right|
        \le
        c\lambda \Phi_I(p,q).
\]
\end{proposition}

\begin{proof}
We prove the proposition by induction on \(|I|\).  If \(|I|=1\), then
\(\widetilde T_I=T_I\), so there is nothing to prove. Assume now that \(I=L\circ R\), and that the claim holds for \(L\) and
\(R\).  We decompose
\[
\begin{aligned}
        \Delta_I
        &\deff
        \widetilde T_I-T_I                                                \\
        &=
        (\widetilde T_I-\widetilde T_L\widetilde T_R)
        +
        \Delta_L T_R
        +
        T_L\Delta_R
        +
        \Delta_L\Delta_R .
\end{aligned}
\]
We bound the four terms separately.

First, by \cref{prop:local-expander-product-error},
\[
        \left|
        p^\top
        (\widetilde T_I-\widetilde T_L\widetilde T_R)
        q
        \right|
        \le
        \lambda F_L(p)B_R(q).
\]
Second, by the induction hypothesis applied to the interval \(L\),
\[
\begin{aligned}
        \left|
        p^\top \Delta_L T_Rq
        \right|
        &\le
        c\lambda \Phi_L(p,T_Rq).
\end{aligned}
\]
Third, by the induction hypothesis applied to the interval \(R\),
\[
\begin{aligned}
        \left|
        p^\top T_L\Delta_Rq
        \right|
        =
        \left|
        (T_L^\top p)^\top \Delta_Rq
        \right|                                      
        \le
        c\lambda \Phi_R(T_L^\top p,q).
\end{aligned}
\]
It remains to bound the second-order term
\(
        p^\top \Delta_L\Delta_Rq.
\)
By Cauchy--Schwarz,
\[
        \left|
        p^\top \Delta_L\Delta_Rq
        \right|
        \le
        \|\Delta_L^\top p\|_2 \cdot \|\Delta_Rq\|_2 .
\]
We first bound \(\|\Delta_L^\top p\|_2\).  By duality,
\[
\begin{aligned}
        \|\Delta_L^\top p\|_2
        &=
        \sup_{\|r\|_2=1}
        \left|
        p^\top\Delta_L r
        \right|                                                     \\
        &\le
        c\lambda
        \sup_{\|r\|_2=1}
        \Phi_L(p,r)                                                 \\
        &\le
        c\lambda
        \sup_{\|r\|_2=1}
        F_L(p)B_L(r)                                                \\
        &\le
        c\lambda F_L(p)\kappa_B(L),
\end{aligned}
\]
where we used \cref{lem:phi-dominated-by-fb} in the third line.
Similarly,
\[
\begin{aligned}
        \|\Delta_Rq\|_2
        &\le
        c\lambda \kappa_F(R)B_R(q).
\end{aligned}
\]
Therefore,
\[
        \left|
        p^\top \Delta_L\Delta_Rq
        \right|
        \le
        c^2\lambda^2
        \kappa_B(L)\kappa_F(R)
        F_L(p)B_R(q).
\]
By the definition of \(K\),
\(
        \kappa_B(L)\kappa_F(R)\le K.
\)
Hence the expander error and the second-order error together are at
most
\[
        \bigl(\lambda+c^2\lambda^2K\bigr)F_L(p)B_R(q)
        \le
        c\lambda F_L(p)B_R(q),
\]
using the assumption \(1+c^2\lambda K\le c\). Combining the four bounds gives
\[
\begin{aligned}
        \left|
        p^\top \Delta_Iq
        \right|
        &\le
        c\lambda F_L(p)B_R(q)
        +
        c\lambda \Phi_L(p,T_Rq)
        +
        c\lambda \Phi_R(T_L^\top p,q)        \\
        &=
        c\lambda \Phi_I(p,q).
\end{aligned}
\]
This completes the induction.
\end{proof}

\begin{remark}
\normalfont
By \cref{lemma:weight-bound} and Cauchy--Schwarz, for every interval \(J\) and every
\(v\in\mathbb R^w\),
\[
\begin{aligned}
        B_J(v)
        &\le
        2\sum_{1\le i<j\le w}|v_i-v_j|                                      \\
        &\le
        2\binom{w}{2}^{1/2}
        \left(
        \sum_{1\le i<j\le w}(v_i-v_j)^2
        \right)^{1/2}                                                        \\
        &=
        2\binom{w}{2}^{1/2}
        \left(
        w\|v\|_2^2-\langle v,\mathbf 1\rangle^2
        \right)^{1/2}                                                        \\
        &\le
        \sqrt 2\, w^{3/2}\|v\|_2.
\end{aligned}
\]
The same bound holds for \(F_J(v)\).  Thus
\[
        \kappa_F(J),\kappa_B(J)\le \sqrt 2\,w^{3/2},
\]
and consequently
\(
        K\le 2w^3.
\)
In particular, the assumption
\(
        \lambda\le \frac{1}{8c^2w^3}
\)
from \cref{prop:error-bound} is more than sufficient for
\(1+c^2\lambda K\le c\), whenever \(c\ge 2\).
\end{remark}

The refined induction immediately implies \cref{prop:error-bound}, since
\cref{lem:phi-dominated-by-fb} gives
\(
        \Phi_I(p,q)\le F_I(p)B_I(q).
\)

\subsection{A Convolution Perspective}
When expanding the recursion underlying the definition of \(\Phi_I\), we see it is exactly an ordered
convolution of the local forward and backward weights, as discussed in this section.

Let \(I=[a,b)\), let \(p_a=p\), and let \(q_b=q\).  Define the propagated
vectors
\[
        p_{t+1}=T_t^\top p_t
        \qquad\text{for }t=a,a+1,\ldots,b-1,
\]
and
\[
        q_t=T_tq_{t+1}
        \qquad\text{for }t=b-1,b-2,\ldots,a.
\]
For \(t\in I\), define the local forward and backward weights
\[
        f_t
        =
        \left\|
        p_t^\top(P_{t,0}-P_{t,1})
        \right\|_1,
        \qquad
        b_t
        =
        \left\|
        (P_{t,0}-P_{t,1})q_{t+1}
        \right\|_1.
\]
Thus
\[
        F_I(p)=\sum_{t=a}^{b-1}f_t,
        \qquad
        B_I(q)=\sum_{t=a}^{b-1}b_t.
\]

\begin{lemma}[Convolution identity]
\label{lem:phi-convolution}
For every dyadic interval \(I=[a,b)\) and every \(p,q\in\mathbb R^w\),
\[
        \Phi_I(p,q)
        =
        \sum_{a\le t<u<b} f_t b_u.
\]
\end{lemma}

\begin{proof}
We prove the claim by induction on \(|I|\).  If \(|I|=1\), then both sides
are zero.

Assume \(I=L\circ R\), where \(L=[a,m)\) and \(R=[m,b)\).  By definition,
\[
        \Phi_I(p,q)
        =
        F_L(p)B_R(q)
        +
        \Phi_L(p,T_Rq)
        +
        \Phi_R(T_L^\top p,q).
\]
The first term is
\[
        F_L(p)B_R(q)
        =
        \left(\sum_{t=a}^{m-1}f_t\right)
        \left(\sum_{u=m}^{b-1}b_u\right)
        =
        \sum_{\substack{a\le t<m\\ m\le u<b}} f_t b_u.
\]
This accounts for all ordered pairs \((t,u)\) separated by the split.

Next, since \(T_Rq=q_m\), the backward propagation inside \(L\) with
terminal vector \(T_Rq\) is exactly the restriction of the above backward
propagation to \(L\).  Therefore, by the induction hypothesis,
\[
        \Phi_L(p,T_Rq)
        =
        \sum_{a\le t<u<m}f_t b_u.
\]
Similarly, since \(T_L^\top p=p_m\), the forward propagation inside \(R\)
with initial vector \(T_L^\top p\) is exactly the restriction of the above
forward propagation to \(R\).  Hence
\[
        \Phi_R(T_L^\top p,q)
        =
        \sum_{m\le t<u<b}f_t b_u.
\]
Adding the three sums gives
\[
        \Phi_I(p,q)
        =
        \sum_{a\le t<u<b} f_t b_u,
\]
as claimed.
\end{proof}

\section{Optimality of the spectral analysis}\label{sec:lower-bound}
\label{sec:spectral-lower-bound}


In this section we consider permutation branching programs
over an alphabet $[d]$. We prove a lower bound of
$
    \Omega\!\left(\log n\cdot\min\{\log d,\log\log n\}\right)
$
on the seed length of the INW generator required by any spectral analysis, even for width-$2$ programs and
constant error. We recall that by "spectral analysis" we mean any analysis of the INW generator which only assumes upper bounds on the spectral expansion parameter.
This lower bound strengthens the alphabet-dependent lower bound of
Hoza, Pyne, and Vadhan~\cite{HWP24}, which has
$\log\log\min\{n,d\}$ in place of $\min\{\log d,\log\log n\}$.
Together with their dependence on width and error,
this gives a lower bound of
\begin{equation}\label{eq:lower-bound-seed}
    \Omega\!\left(\log n\cdot\log(w/\varepsilon)
    +\log n\cdot\min\{\log d,\log\log n\}
    +\log d\right).
\end{equation}

Combined with the extension of our analysis to general alphabet (\cref{thm:our-with-general-d}) and with the upper bound of~\cite{HPV21}, this determines the optimal seed length obtainable by
any spectral analysis of the INW generator for permutation BPs.

We recall the upper bound of Hoza, Pyne, and Vadhan~\cite{HPV21}. They analyzed the INW generator for \emph{unbounded-width} permutation BPs,
with a single accepting state. We formulate their result in the setting of permutation BPs of width $w$ and arbitrary sets of accepting states.

\begin{theorem}[\cite{HPV21}, Theorem 1.4]\label{thm:HPV21}
    For all $n,d,w \geq 2$ and for all $0 < \varepsilon \leq 1$, there exists an expander-based INW-generator
    that $\varepsilon$-fools length-$n$ width-$w$ permutation BPs over $[d]$. This PRG has seed length
    \[
        O(\log d + \log n \cdot (\log\log n + \log(1/\varepsilon) + \log(w))).
    \]
\end{theorem}

Consequently, taking the minimum between \cref{thm:our-with-general-d} and \cref{thm:HPV21} exactly matches the lower bound in \cref{eq:lower-bound-seed}.

\subsection{Expander construction}
\label{subsec:biased-parity-expander}
We first construct an expander with a biased parity property, which will be used to construct a weak INW generator.
\begin{lemma}\label{lem:biased-parity-expander}
    Let $\Omega,\omega:V\to\{0,1\}$ with
    $0<\mathbb E[\omega]\leq\mathbb E[\Omega]\leq1/2$.
    Then there is an expander $H=(V,E)$ such that
    $\lambda(H)\leq\sqrt{2\mathbb E[\omega]/\mathbb E[\Omega]}$ and
    \[
        \mathbb E_{(x,y)\in E}[\Omega(x)\oplus\omega(y)]
        =\mathbb E[\Omega].
    \]
\end{lemma}

\begin{proof}
    Let $N=|V|$, $p=\mathbb E[\omega]$, and $P=\mathbb E[\Omega]$.
    We think of $\Omega,\omega$ as vectors in $\mathbb R^N$.
    Let $\mathbf{1}\in\mathbb R^N$ denote the all-ones vector.
    Define $K\in\mathbb R^{N\times N}$ by the formula
    \[
        K=J+\frac{(\Omega-P\mathbf{1})(\omega-p\mathbf{1})^{\top}}
                     {N\cdot P\cdot(1-p)},
    \]
    where $J$ is the normalized all-ones matrix. Then every entry of $K$ is
    nonnegative, because $p\leq P$ and
    \[
        K_{y,x}\geq\frac1N-\frac{P\cdot(1-p)}{N\cdot P\cdot(1-p)}=0.
    \]
    Furthermore, $K$ is doubly stochastic, because $J$ is doubly
    stochastic and
    \[
        \sum_x(\Omega(x)-P)=\sum_y(\omega(y)-p)=0.
    \]
    The matrix $K$ is not necessarily symmetric, so we symmetrize it:
    \[
        H_0=\frac{K+K^{\top}}2.
    \]
    To bound the spectral expansion parameter, observe that
    \[
        \lambda(H_0)
        \leq\frac12\bigl(\|K-J\|_{\mathrm{op}}
                         +\|K^{\top}-J\|_{\mathrm{op}}\bigr)
        =\frac{\|\Omega-P\mathbf{1}\|_2\cdot\|\omega-p\mathbf{1}\|_2}
                    {N\cdot P\cdot(1-p)}.
    \]
    We have
    \[
        \begin{aligned}
            \|\Omega-P\mathbf{1}\|_2^2
            &=\langle\Omega-P\mathbf{1},\Omega-P\mathbf{1}\rangle\\
            &=PN-2P^2N+P^2N=P\cdot(1-P)\cdot N,
        \end{aligned}
    \]
    and similarly
    \[
        \|\omega-p\mathbf{1}\|_2^2=p\cdot(1-p)\cdot N,
    \]
    so
    \[
        \begin{aligned}
            \lambda(H_0)
            &\leq\frac{\sqrt{P\cdot(1-P)\cdot N}
                         \cdot\sqrt{p\cdot(1-p)\cdot N}}
                        {N\cdot P\cdot(1-p)}\\
            &=\sqrt{\frac{p\cdot(1-P)}{P\cdot(1-p)}}
             \leq\sqrt{2p/P}.
        \end{aligned}
    \]
    Next, we will bound
    $\mathbb E_{(x,y)\in E(H_0)}[\Omega(x)\oplus\omega(y)]$.
    The first step is to compute
    \[
        \begin{aligned}
            \mathbb E_{(x,y)\in E(K)}[\Omega(x)\cdot\omega(y)]
            &=\frac{\Omega^{\top}K\omega}{N}\\
            &=Pp+\frac{\Omega^{\top}(\Omega-P\mathbf{1})
                          \cdot(\omega-p\mathbf{1})^{\top}\omega}
                         {N^2\cdot P\cdot(1-p)}\\
            &=Pp+\frac{(PN-P^2N)\cdot(pN-p^2N)}
                         {N^2\cdot P\cdot(1-p)}\\
            &=Pp+p\cdot(1-P).
        \end{aligned}
    \]
    Next, we compute
    \[
        \begin{aligned}
            \mathbb E_{(x,y)\in E(K)}[\Omega(y)\cdot\omega(x)]
            &=\frac{\omega^{\top}K\Omega}{N}\\
            &=Pp+\frac{\omega^{\top}(\Omega-P\mathbf{1})
                          \cdot(\omega-p\mathbf{1})^{\top}\Omega}
                         {N^2\cdot P\cdot(1-p)}\\
            &=Pp+\frac{(\omega^{\top}\Omega-NPp)^2}
                         {N^2\cdot P\cdot(1-p)}\\
            &\geq Pp.
        \end{aligned}
    \]
    Here we have used that $\mathbf{1}^{\top}\Omega=NP$ and
    $\mathbf{1}^{\top}\omega=Np$.
    Therefore,
    \[
        \mathbb E_{(x,y)\in E(H_0)}[\Omega(x)\cdot\omega(y)]
        \geq Pp+\frac{p\cdot(1-P)}2.
    \]
    Consequently,
    \[
        \begin{aligned}
            &\mathbb E_{(x,y)\in E(H_0)}[\Omega(x)\oplus\omega(y)]\\
            &\quad=\mathbb E[\Omega(x)\cdot(1-\omega(y))]
                    +\mathbb E[(1-\Omega(x))\cdot\omega(y)]\\
            &\quad=\mathbb E[\Omega]-\mathbb E[\Omega(x)\cdot\omega(y)]
                    +\mathbb E[\omega]-\mathbb E[\Omega(x)\cdot\omega(y)]\\
            &\quad=P+p-2\mathbb E[\Omega(x)\cdot\omega(y)]\\
            &\quad\leq P+p-2\left(Pp+\frac{p\cdot(1-P)}2\right)\\
            &\quad=P+p-2Pp-p+Pp\\
            &\quad=P-Pp\\
            &\quad\leq P.
        \end{aligned}
    \]
    On the other hand, if we pick vertices independently, then we have
    \[
        \mathbb E_{x,y\in V}[\Omega(x)\oplus\omega(y)]\geq P,
    \]
    because $P\leq1/2$ and XORing two biased coins decreases their
    bias. Indeed, we have
    \[
        \begin{aligned}
            \left|\mathbb E_{x,y\in V}
                [(-1)^{\Omega(x)}\cdot(-1)^{\omega(y)}]\right|
            &=\left|\mathbb E_x[(-1)^{\Omega(x)}]\right|
                \cdot\left|\mathbb E_y[(-1)^{\omega(y)}]\right|\\
            &\leq\left|\mathbb E_x[(-1)^{\Omega(x)}]\right|.
        \end{aligned}
    \]
    Therefore, there is some $\alpha\in[0,1]$ such that the matrix
    $H=(1-\alpha)\cdot H_0+\alpha\cdot J$ satisfies all the
    conditions in the lemma. Furthermore, since the entries of $H_0$
    are rational, the parameter $\alpha$ can be chosen rational as
    well, hence $H$ is the transition probability matrix of a regular
    undirected multigraph.
\end{proof}

\subsection{Using the expander to construct a weak INW generator}
\label{subsec:alphabet-lower-bound}

In this subsection we prove the lower bound on the seed length of the INW generator required by any spectral analysis.
We begin by presenting the following notation.

\begin{definition}[\cite{HWP24}, definition 2.5]\label{def:inw-spectral-analysis}
    For $d_0\in\mathbb N$ and $\lambda_1,\ldots,\lambda_\ell\geq0$,
    let $\textrm{INW}(d_0,\lambda_1,\ldots,\lambda_\ell)$ be the set
    of INW PRGs $\textrm{GEN}:[S]\rightarrow[d_0]^{2^\ell}$ constructed
    with auxiliary undirected regular graphs $H_1,\ldots,H_\ell$ where
    $\lambda(H_i)\leq\lambda_i$ for all $i$.
    We say $\textrm{INW}(d_0,\lambda_1,\ldots,\lambda_\ell)$
    $\varepsilon$-fools a class of functions $\mathcal F$ if every
    $\textrm{GEN}\in\textrm{INW}(d_0,\lambda_1,\ldots,\lambda_\ell)$
    $\varepsilon$-fools every $f\in\mathcal F$.
    Furthermore, define $s_{\textrm{INW}}(d_0,\lambda_1,\ldots,\lambda_\ell)$
    as the minimal seed length of all PRGs in
    $\textrm{INW}(d_0,\lambda_1,\ldots,\lambda_\ell)$.
    We call the set $(\lambda_1,\ldots,\lambda_\ell)$ a constraint,
    and say a family of graphs $(H_1,\ldots,H_\ell)$ satisfies the
    constraint if $\lambda(H_i)\leq\lambda_i$ for all $i$.
\end{definition}

\begin{proposition}\label{thm:alphabet-spectral-constraints}
    Let $d \geq 2$, let $n$ be a power of two, and let
    $\vec\lambda\in[0,1]^{\log n}$. Let $k=\min\{d,\log n\}$.
    If every generator in $\mathrm{INW}(d,\vec\lambda)$ fools width-$2$
    permutation ROBPs over the alphabet $[d]$ with error $0.01$,
    then there are $\Omega(\log n)$ values of $i$ such that
    $\lambda_i\leq O(1/\sqrt{k})$.
\end{proposition}

\begin{proof}
    We prove the contrapositive. Let $I$ be the set of $i$ such that
    $\lambda_i>\sqrt{200/k}$, and suppose $|I|\geq0.5\log n$.
    Under this assumption, we will construct a generator
    $G\in\mathrm{INW}(\vec\lambda)$ and a branching program $f:[d]^n\rightarrow\{0,1\}$
    such that $G$ does not fool $f$.

    The first layer of $f$ swaps with probability $0.01$.
    (Assume for simplicity that $d$ is a multiple of $100$.)
    Now assume by induction that we have constructed the first $2^i$
    layers of $f$ and we have also constructed the INW generators
    $G_0,G_1,\ldots,G_i$, where $G_i$ outputs $2^i$ symbols.
    Now let us double the lengths of the program and the generator.

    \begin{itemize}
        \item If $i+1\notin I$, then the next $2^i$ layers of $f$
        are identity layers, and we use the complete graph to construct
        $G_{i+1}$: $G_{i+1}(x,y)=(G_i(x),G_i(y))$.

        \item If $i+1\in I$, then the next $2^i$ layers of $f$ are
        identity layers except for one layer that swaps with probability
        $1/k$ (assume for simplicity that $d$ is a multiple of $k$).
        Let $\Omega$ be the event that the first $2^i$ steps do an odd
        number of swaps, and let $\omega$ be the event that the next
        $2^i$ steps do a swap. We use the expander $H$ given by \cref{lem:biased-parity-expander}
        to construct $G_{i+1}$. By induction, $\mathbb E[\Omega]=0.01$,
        hence $\lambda(H)\leq\sqrt{200/k}$.
    \end{itemize}

    By construction, under the INW generator, $f$ accepts with probability $0.01$. On the other hand, under uniform randomness, we have
    \[
        \begin{aligned}
            |\mathbb E[(-1)^f]|
            &=|1\cdot(1-1/k)-1\cdot1/k|^{|I|}
             =(1-2/k)^{|I|}\\
            &\leq(1-2/k)^{\frac12\log n}
             \leq e^{-(\log n)/k}\leq1/e,
        \end{aligned}
    \]
    hence
    \[
        \mathbb E[f]=\frac12-\frac12\mathbb E[(-1)^f]
        \geq\frac12-\frac1{2e}>0.3.
    \]
\end{proof}

The following theorem readily follows from
\cref{thm:alphabet-spectral-constraints}.

\begin{proposition}\label{thm:alphabet-spectral-lower-bound}
    Let $d\geq2$, $n=2^h\geq4$, and $\lambda_1,\ldots,\lambda_h\geq0$.
    If $\mathsf{INW}(d,\lambda_1,\ldots,\lambda_h)$ $0.01$-fools all
    length-$n$ width-$2$ permutation BPs over $[d]$, then every
    generator in this family has seed length
    \[
        \Omega\!\left(h\log\min\{d,h\}\right)
        =\Omega\!\left(\log n\cdot
                    \min\{\log d,\log\log n\}\right).
    \]
\end{proposition}

\begin{theorem}\label{thm:combined-spectral-lower-bound}
    Fix $0<\delta<1$. For sufficiently large $n=2^h$, let
    $d,w\geq2$, $0<\varepsilon\leq0.01$, and
    $\lambda_1,\ldots,\lambda_h\geq0$, with
    $\log(w/\varepsilon)\leq n^{1-\delta}$.
    If $\mathsf{INW}(d,\lambda_1,\ldots,\lambda_h)$
    $\varepsilon$-fools all length-$n$ width-$w$ permutation BPs
    over $[d]$ with arbitrary sets of accepting states, then every
    generator in this family has seed length
    \[
        \Omega_\delta\!\left(
            \log n\cdot\log(w/\varepsilon)
            +\log n\cdot\min\{\log d,\log\log n\}
            +\log d
        \right).
    \]
    Moreover, this lower bound is tight up to a constant factor.
\end{theorem}

\begin{proof}
    By~\cite[Theorems~5.2 and~3.1]{HWP24}, the seed length is at least
    \[
        \Omega\!\left(\log w\cdot(\log n-\log\log w)\right)
    \]
    and
    \[
        \Omega\!\left(\log(1/\varepsilon)\cdot
            (\log n-\log\log(1/\varepsilon))\right),
    \]
    respectively. Since $\log(w/\varepsilon)\leq n^{1-\delta}$, we have
    $\log\log(w/\varepsilon)\leq(1-\delta)\log n$, so the seed length is at least
    $\Omega_\delta(\log n\cdot\log(w/\varepsilon))$.
    Since $\varepsilon\leq 0.01$,
    \cref{thm:alphabet-spectral-lower-bound} implies
    $\Omega(\log n\cdot\min\{\log d,\log\log n\})$.
    Finally, every INW generator has seed length
    at least $\log d$, which completes the proof of the lower bound.

    This bound is tight up to a constant factor, as shown by the upper bound of \cref{thm:our-with-general-d} and \cref{thm:HPV21}.
\end{proof}

\section{Acknowledgment and AI Disclosure}\label{sec:ack}
The authors used AI tools solely for copy-editing. Except for \cref{sec:lower-bound}, all original ideas and proofs were developed by the authors without AI assistance.

\cref{sec:lower-bound} was mostly written by William Hoza, who obtained the proofs through interactions with ChatGPT. We thank William for communicating the result to us and allowing us to incorporate it in our work (and take full responsibility for the correctness of the claims). We also thank Salil Vadhan for his very insightful feedback on an earlier version of this work.

\bibliographystyle{alpha}
\bibliography{bib}

\appendix

\section{Missing Proofs}\label{sec:missing proofs}

For completeness, in this section we include some missing proofs.
We begin with \cref{lemma:expander-mixing-vectors}, which, as mentioned, follows from the standard expander-mixing lemma.

\begin{proof}[Proof of \cref{lemma:expander-mixing-vectors}]
    Write $\alpha = (\alpha_1,\ldots, \alpha_k), \beta = (\beta_1,\ldots,\beta_k) : V \to \mathbb{R}^k$. By the standard expander mixing lemma, for all $i$ we have
    \[
        \left | \mathbb{E}_{(x,y) \in E}[ \alpha_i(x) \beta_i(y)]\right | \leq \lambda \left (\mathbb{E}_{x \in V} [\alpha_i(x)^2] \right )^{1/2} \left (\mathbb{E}_{y \in V} [\beta_i(y)^2 ] \right)^{1/2}.
    \]
    Thus,
    \begin{align*}
        \left | \mathbb{E}_{(x,y) \in E} [\langle \alpha(x), \beta(y) \rangle_{\mathbb{R}^k}] \right | 
        &= \left | \mathbb{E}_{(x,y) \in E} \left[ \sum_{i=1}^k \alpha_i(x) \beta_i(y) \right] \right | \\
        &\leq \sum_{i=1}^k \left | \mathbb{E}_{(x,y) \in E} [ \alpha_i(x) \beta_i(y) ] \right | \\
        &\leq \lambda \sum_{i=1}^k \left (\mathbb{E}_{x \in V} \left [ \alpha_i(x)^2 \right ] \right )^{1/2} \left (  \mathbb{E}_{y \in V} \left [ \beta_i(y)^2 \right ] \right )^{1/2}.
    \end{align*}
    By the Cauchy--Schwarz inequality, the above is bounded by
    \[
        \lambda \left ( \EE_{x \in V} \norm{\alpha(x)}_2^2) \right )^{1/2} \left ( \EE_{y \in V} \norm{\beta(y)}_2^2) \right )^{1/2},
    \]
    as needed.
\end{proof}

We include the proof of \cref{lemma:weight-bound}. As noted above, the proof follows directly from Lemma~5 of \cite{BRRY}.

\begin{proof}[Proof of \cref{lemma:weight-bound}]
    We first prove the bound for $B_I(q)$. Consider the subprogram induced by
    the interval $I=[a,b)$, and label the vertices in layer $b$ by the entries
    of $q=q_b$. Extend these labels backwards by setting
    \[
        q_t
        =
        T_t q_{t+1}
        =
        \frac{P_{t,0}+P_{t,1}}{2}q_{t+1}
    \]
    for $t=b-1,b-2,\ldots,a$. Thus the label of each vertex is the average of
    the labels of its two children. Since the branching program is a permutation
    branching program, this subprogram is regular.

    The weight of this evaluation program, in the sense of \cite{BRRY}, is
    exactly
    \[
        \sum_{t=a}^{b-1}
        \left(
            \norm{P_{t,0}q_{t+1}-q_t}_1
            +
            \norm{P_{t,1}q_{t+1}-q_t}_1
        \right)
        =
        B_I(q).
    \]
    Lemma~5 of \cite{BRRY} therefore gives
    \[
        B_I(q)
        \leq
        2\sum_{1\leq i<j\leq w}|q_i-q_j|.
    \]
    Although Lemma~5 of \cite{BRRY} is stated for labels in $[0,1]$, this
    entails the displayed bound for arbitrary $q\in\mathbb{R}^w$ by an affine
    rescaling. Indeed, if $q$ is not constant, set
    \[
        \widehat q
        =
        \frac{q-\alpha\mathbbm{1}}{\beta-\alpha},
        \qquad
        \alpha=\min_i q_i,\quad \beta=\max_i q_i.
    \]
    Since $T_t\mathbbm{1}=\mathbbm{1}$ for every $t$, the corresponding
    propagated labels satisfy
    \[
        \widehat q_t
        =
        \frac{q_t-\alpha\mathbbm{1}}{\beta-\alpha}.
    \]
    Hence all edge weights, and also all pairwise differences in the terminal
    layer, are scaled by the same factor $\beta-\alpha$. Applying the
    $[0,1]$ version to $\widehat q$ and scaling back gives the desired bound.
    If $q$ is constant, then both sides are zero.

    The bound for $F_I(p)$ follows by applying the same argument to the
    reversed permutation branching program on the interval $I$. In the reversed
    program, the transition matrices are $P_{t,0}^\top$ and $P_{t,1}^\top$.
    Since the matrices $P_{t,0}$ and $P_{t,1}$ are permutation matrices, their
    transposes are their inverses and are again permutation matrices. The
    backward propagation in the reversed program is precisely
    \[
        p_{t+1}
        =
        T_t^\top p_t
        =
        \frac{P_{t,0}^\top+P_{t,1}^\top}{2}p_t,
    \]
    which is the forward propagation in the original program. Moreover, the
    corresponding edge weight at layer $t$ is
    \[
        \norm{P_{t,0}^\top p_t-p_{t+1}}_1
        +
        \norm{P_{t,1}^\top p_t-p_{t+1}}_1
        =
        F_t(p).
    \]
    Therefore Lemma~5 of \cite{BRRY}, applied to the reversed program, gives
    \[
        F_I(p)
        \leq
        2\sum_{1\leq i<j\leq w}|p_i-p_j|.
    \]

    Finally, for any $v\in\mathbb{R}^w$,
    \[
        \sum_{1\leq i<j\leq w}|v_i-v_j|
        \leq
        \sum_{1\leq i<j\leq w}(|v_i|+|v_j|)
        =
        (w-1)\norm{v}_1
        \leq
        w\norm{v}_1.
    \]
    Applying this with $v=q$ and $v=p$ gives the two remaining inequalities.
\end{proof}

\end{document}